\documentclass[10pt,a4paper]{article}

\usepackage{latexsym,amssymb,amsmath,amsthm} 
\usepackage{graphicx} 
\usepackage{amssymb}
\usepackage{setspace}
\usepackage{color}
\usepackage{courier}
\definecolor{mygray}{rgb}{0.8,0.8,0.8}
\usepackage{listings}
\usepackage{multirow}
\usepackage{algorithm}
\usepackage[noend]{algorithmic}
\usepackage{array}
\usepackage{subfigure}
\usepackage{subfigure}

\newtheorem{theorem}{Theorem}[section]
\newtheorem{lemma}[theorem]{Lemma}

\begin{document}

\title{Scheduling Chained Multiprocessor Tasks onto Large Multiprocessor System}

\author{T. K. Agrawal, R, Sharma, M. Ghose, A. Sahu \ \\ Department of Comp. Sc \& Engg., IIT Guwahati, \\
Guwahati, Assam, India, PIN-781039. email:\{tarun.k, r.sharma, g.manojit, asahu\}@iitg.ernet.in\\[-2pt]}
\date{}

\maketitle

\begin{abstract}
In this paper, we proposed an effective approach for scheduling of multiprocessor unit time tasks with chain precedence on to large multiprocessor system. The proposed longest chain maximum processor scheduling algorithm is proved to be optimal for uniform chains and monotone (non-increasing/non-decreasing) chains for both splitable and non-splitable multiprocessor unit time tasks chain. Scheduling arbitrary chains of non-splitable multiprocessor unit time tasks is proved to be NP-complete problem. But scheduling arbitrary chains of splitable multiprocessor unit time tasks is still an open problem to be proved whether it is NP-complete or can be solved in polynomial time. We have used three heuristics (a) maximum criticality first, (b) longest chain maximum criticality first and (c) longest chain maximum processor first for scheduling of arbitrary chains. Also compared performance of all three scheduling heuristics and found out that the proposed longest chain maximum processor first performs better in most of the cases. 
\end{abstract}

\section{Introduction and Motivation} 

Modern computing system contains multiple cores which enable many applications or tasks to execute concurrently. Chip multiprocessor exploits the increasing device density in a single chip, so now a days we expects two or three order number of cores on a chip. Also most of the applications by nature they are parallel and their run time characteristics exhibit time varying phase behavior \cite{Sherwood20, Cho21}. Applications impose different performance metric values in different phases. During a phase an application have same value of performance metrics. In  \cite{Ban22,SureshIndicon}, Banerjee $et~al.$ used instruction per cycle (IPC), instruction level parallelism (ILP) and L1 cache hits to detect phases of applications execution time. As different phases of application have different parallelism and memory access characteristics scheduling algorithm should consider this fact to improve the performance, otherwise system will be underutilized. In this work, we have considered the parallelism characteristics of different phases of an application to schedule efficiently. Without loss of generality, we can consider an application consist of a sequence of task or phase, where each task or phase exhibit different degree of parallelism. In this text, we have used task and phase interchangeably. Tasks that require more than one processor at a time are called multiprocessor tasks. Tasks requiring $k$ processors at a time are called $k$-width tasks. 

In this paper we are concerned about scheduling of $N$ multi-phase applications (chain of multiprocessor tasks, each of which have arbitrary number of phases or tasks) onto $M$ processor system. Each phase or task has two characteristics: one is execution time of phase and other is number of processors required for execution of that phase. A phase of an application can be scheduled on any subset of processors of given size (given by phase). So this kind of application can be represented as collection of multiprocessor task with chain precedence constraints. Interchangeably we say this as chain of multiprocessor task throughout this paper. If execution time of a multiprocessor task is 1 then we say it is a  multiprocessor unit time task and for multiprocessor unit time  task we don't use pre-exemption. A task is splitable multiprocessor unit time task with k-width, means this task may be run in splitable in term of processor. Suppose a task require $1$ unit time $p$ processor then this task may execute $1$ unit on $d$ processor in current time slot and remaining $p-d$ unit in any other time slots. If $d$ is integer then task  is unit splitable otherwise continuous splitable task. In this paper, we have tried to link theoretical aspect of multiprocessor scheduling and efficient scheduling multi-phase application on to multiprocessors to cope up with modern days execution environment scenario (parallel multi-phase application on large multi-processor system). 

Non-preemptive scheduling of $N$ independent tasks (uni-processor task) on $M \geq 3$ processors is NP-complete. Similarly, the problem of multiprocessor tasks scheduling is NP-Hard for non-preemptive case and independent tasks of arbitrary execution time \cite{Blazewicz13}. If the execution time of each phase is restricted to unit time then the problem is polynomially solvable for arbitrary but fixed number of processor requirement case \cite{Blazewicz13}.  When preemption is allowed we can solve independent tasks of arbitrary execution time with arbitrary number of processors requirements in polynomial time \cite{Blazewicz13}. Most of the cases, if we add precedence constraint between tasks difficulty increases. Scheduling multiprocessor applications which have precedence constraint as directed acyclic graphs is NP-Complete. Many heuristics are proposed to schedule them. We are considering precedence constraint chain in this paper which is simplest precedence constraint. The problem scheduling multiprocessor task with chain precedence is also strongly NP-Hard for than two processors and non pre-emptable tasks of arbitrary execution time \cite{Blazewicz13}. If execution time is restricted to unit for chained applications then also the problem is strongly NP-Hard for three processors with arbitrary processors requirement of multiprocessor task \cite{Blazewicz13}. 

In this paper we are considering three types of chains of multiprocessor unit time task. These are (a) uniform chain, (b) monotonically increasing or decreasing chains and (c) arbitrary chains. We are considering two types of multiprocessor tasks also. These are (a) splitable tasks, and (b) non-splitable tasks. Non-splitable task require all required processors simultaneously at a time means $k$-width task can not be scheduled on $<$ $k$ processors and splitable task can be processed by allocating partial number of processors at different times means $k$-width task can be scheduled on processors less than $k$ at a time and remaining can be given on next time. Our result is stronger result as compared to Blazewicz $et~al.$ \cite{5}, where they consider scheduling of tasks requiring an arbitrary number of processors between $1$ and $k$, where $k$ is fixed integer and unit time processing. Also we have considered splitable and non-splitable version of the problem.

Rest of the paper is organized as follows: We have described the problem formulation and variation of problem in Section \ref{problemfor}. We have described previous work in Section \ref{prevwork}. We have described our proposed algorithm for scheduling of uniform chains of  splitable and non-splitable multiprocessor unit time tasks in Section \ref{unichain}. Similarly, we have described algorithm for scheduling of monotone chains of splitable and non-splitable multiprocessor unit time tasks in Section \ref{monotone}. Section \ref{arbichain} describes about compared three heuristics to solve the problem of scheduling arbitrary multiprocessor task chains and evaluates their performance on various scenarios for both splitable and non-splitable task chains. Finally, we have concluded about paper and pointed future works in Section \ref{concl}.

\section{Problem Formulation} \label{problemfor}  

\subsection{Scheduling of multiprocessor tasks on $M$ processor system with chain constraint}

A collection of $N$ application $C$ = $C_{1}$,$C_{2}$,....,$C_{N}$ has to be executed by $M$ identical processors. Each application or chain consists of $n_i$ phases or tasks where $i \in [1,2,..N]$. The processors requirement of task $T_{ij}$ is  $p_{ij}$, where $T_{ij}$ is $j^{th}$ phase or task of application $C_i$ and it satisfies $1 < p_{ij} \leq M$. Execution time $t_{ij}$ of each tasks $T_{ij}$ may be arbitrary. Figure \ref{figure 1} shows an example of application system. In this example, we have 4 applications or chains ($C_{1}$, $C_{2}$, $C_{3}$ and $C_{4}$). Application $C_{1}$, $C_{2}$, $C_{3}$ and $C_{4}$ have 4 phases, 3 phases, 5 phases and 4 phases respectively. A task of an application can't start execution before complete execution of its predecessor task of the same application.

An optimization criterion of multiprocessor scheduling is minimizing makespan time  $C_{max}$. The makespan is defined as the total length of the schedule i.e. when all tasks of all applications are finished i.e. $C_{max} = \smash{\displaystyle\max_{0 \leq i \leq N}} \{F_i\} $ where $F_i$ is the finishing time of $i^{th} $chain or application. 

\begin{figure}[tb]%
\begin{center}%
\subfigure[Arbitrary time]{%
\label{figure 1}%
\includegraphics[scale=0.38]{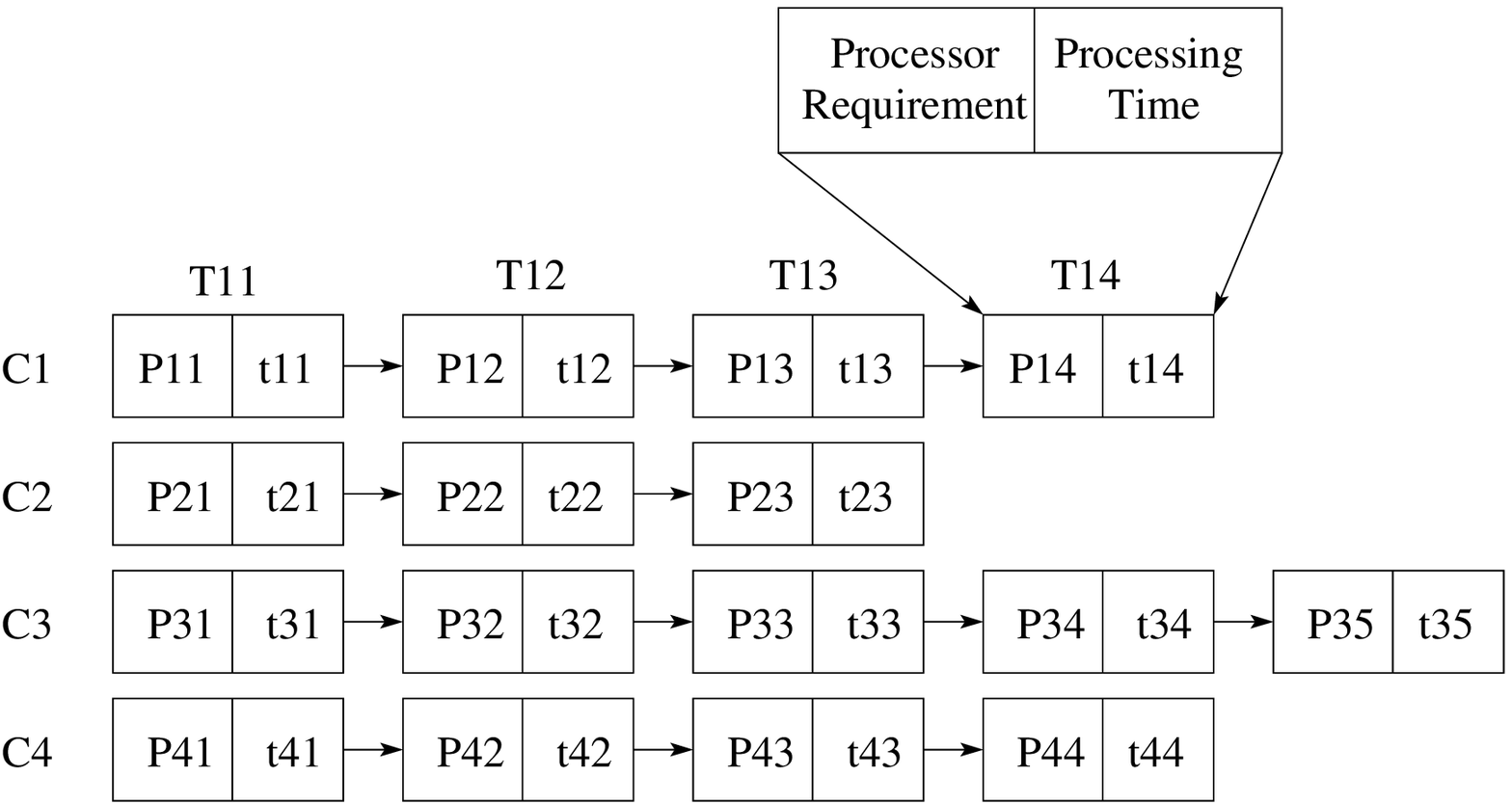} 
}
\hspace{-0.3cm}
\subfigure[Unit time]{%
\label{figure 2}%
\includegraphics[scale=0.38]{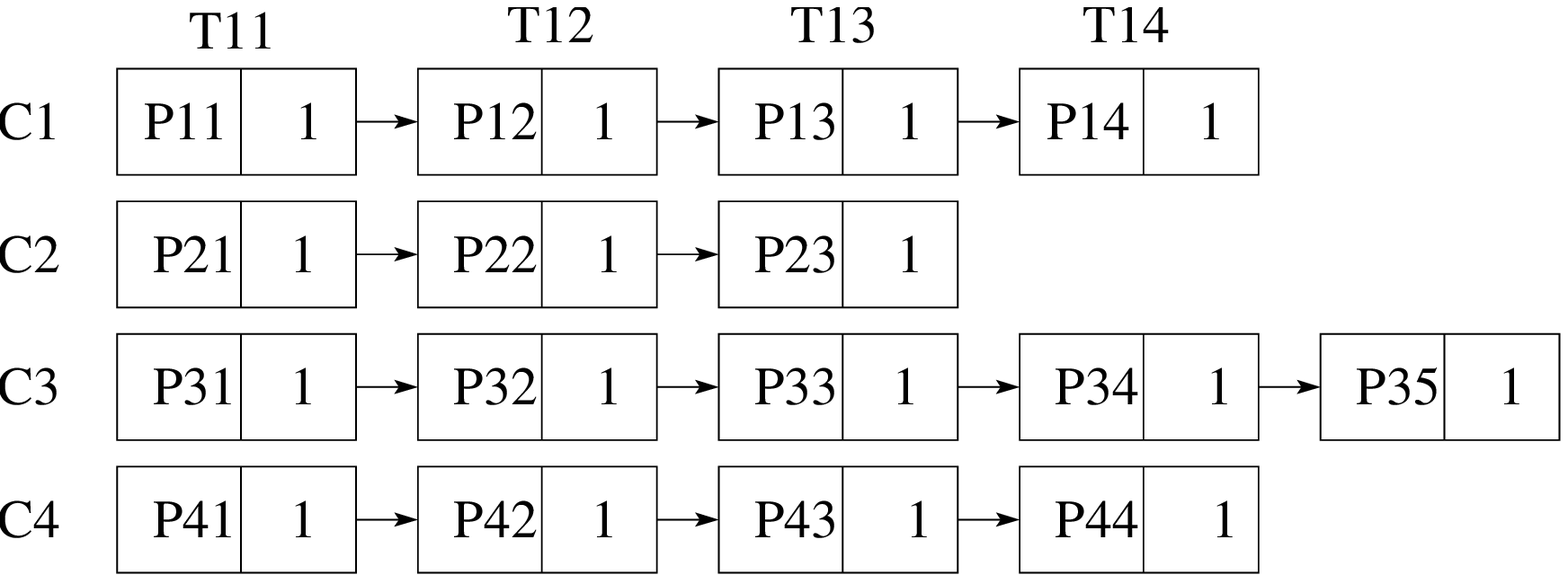} 
}\end{center}
\caption{\textbf{Multiprocessor tasks with chain precedence}}
\label{mptasks}
\end{figure}

The problem of scheduling of multiprocessors arbitrary time tasks with chains is NP-complete, so for simplicity we assume execution time of each tasks $T_{ij}$ equal to unit length i.e. $t_{ij}$ = $1$ where $i$ = $1$ to $N$ and $j$ = $1$ to $n_{i}$, in this paper. We also assume that there is no communication delay between tasks of an application and among the  applications. In the paper, we have used $n$ as total number of tasks (which is different from $n_i$, the number of phases of application $C_i$) and $m$ or $M$ as number of processor in the system.

\subsection{Considered types of multiprocessor tasks}

In this paper, we have assumed two types of tasks: (a) Non-splitable tasks i.e. task can only be processed when all the required number of processor by task are allocated to task at a time and (b) splitable tasks i.e. task can be processed by allocating all required number of processor in pieces at different time. Clearly, we can categorize the  multiprocessor unit time tasks with chain precedence into three following cases: 
\begin{enumerate}
\item \textbf{Uniform chains:} All the tasks of a chain have same number of processors requirement and tasks of different chains may have different processors requirement. Example of this kind of task system is shown in Figure \ref{figuniform}. For a given chain, all the task have same number of parallelism or processor requirement, but this may be different for different chain. 

\begin{figure}[tb]%
\begin{center}%

\subfigure[Uniform]{%
\label{figuniform}%
\includegraphics[scale=0.38]{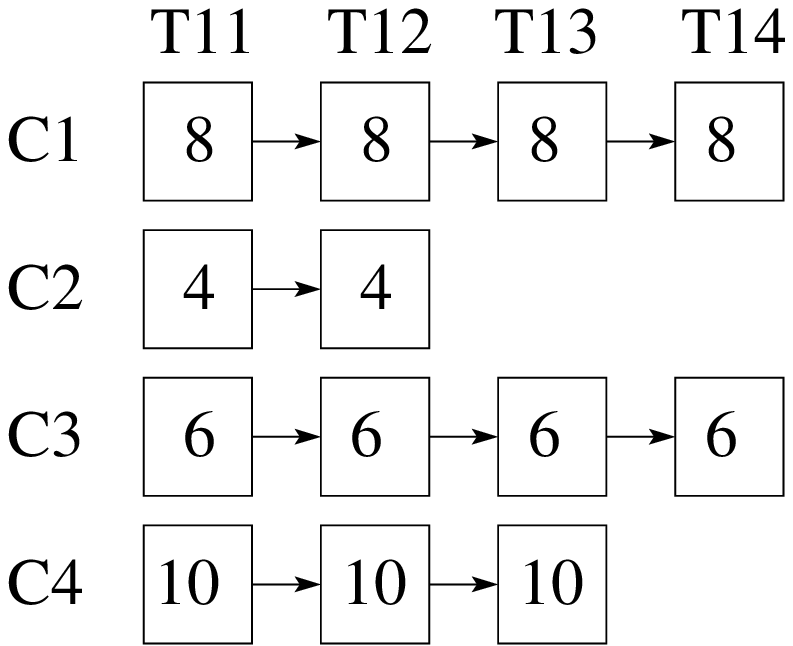} 
}
\subfigure[Decreasing]{%
\label{figdecreasing}%
\includegraphics[scale=0.38]{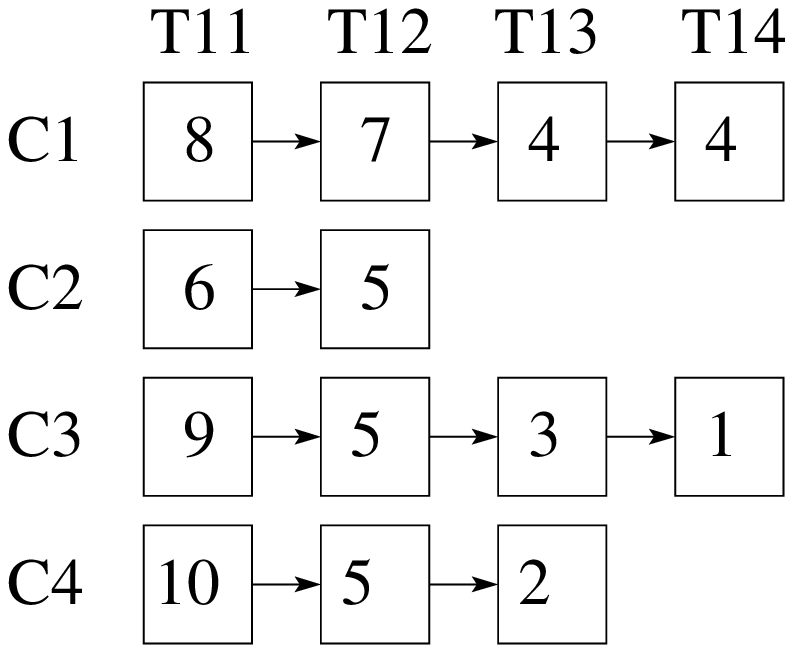} 
}
\subfigure[Increasing]{%
\label{figincreasing}%
\includegraphics[scale=0.38]{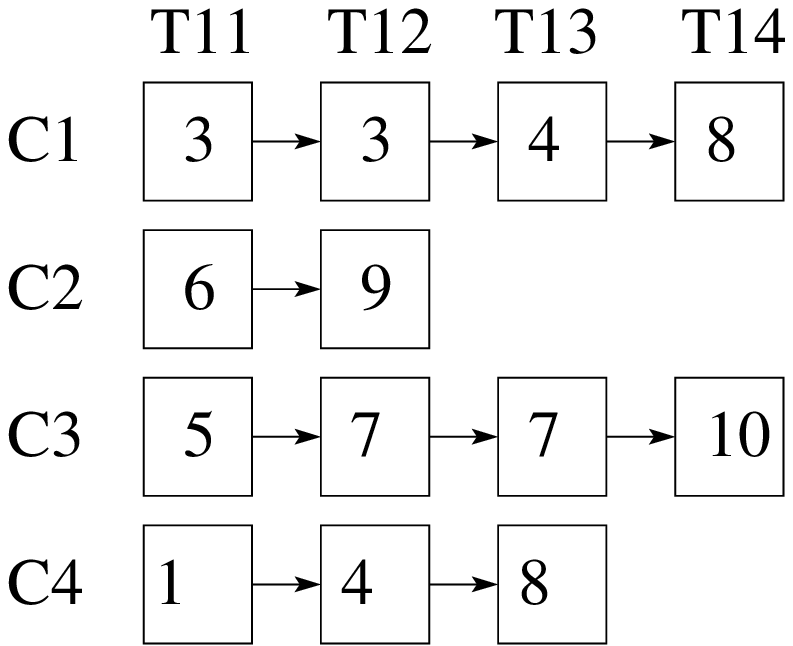} 
}
\subfigure[Arbitrary]{%
\label{figarbitrary}%
\includegraphics[scale=0.38]{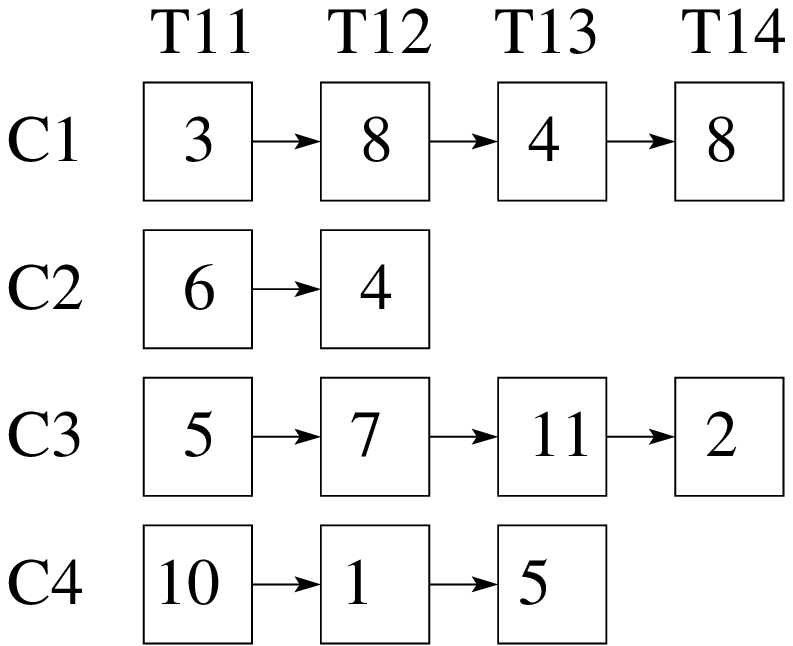} 
}
\end{center}
\caption{\textbf{Multiprocessor unit time tasks with chain precedence}}
\label{umptasks1}
\end{figure}

\item \textbf{Monotone chains:} All the tasks of a chain have non-increasing (or non-decreasing) processor requirement. All the chains of task system are one type of chain either non-increasing or non-decreasing. Figure \ref{figdecreasing} and \ref{figincreasing} shows non-increasing and non-decreasing monotone chains of multiprocessor tasks respectively. Considered monotone task system does not contain mix of both non-increasing and non-decreasing chain.

\item \textbf{Arbitrary Chains:} In this case, tasks of a chain have any arbitrary processors requirement in arbitrary order as shown in Figure \ref{figarbitrary}. 

\end{enumerate}

We have discussed scheduling approach of all three types of chains of multiprocessor tasks (uniform, monotone and arbitrary) and also with both types of multiprocessor tasks (splitable and non-splitable). So it becomes in total six different types of chain of multiprocessor task system. We know that lower bound (LB) of makespan time of scheduling chains of multiprocessor task on multiprocessor can be calculated as \cite{23}
\begin{equation}
LB = \max \Big( \frac{\sum_{i=1}^{N}\sum_{j=1}^{n_i} p_{ij}}{M}, max\{n_i\} \Big)
\end{equation}
where $p_{ij}$ is processor requirement of $j^{th}$ task of $i^{th}$ application, $N$ is number of chain, $M$ is number of processor and $n_i$ is number of phase/task of $i^{th}$ chain. Let us assume OPT is optimal makespan time produced by an optimal algorithm. So makespan time of any arbitrary algorithm will be $LB + P_{waste}$ and it will satisfy the following relation.
\begin{equation}
LB \leq OPT \leq LB + p_{waste}
\end{equation} 
where $p_{waste}$ is minimum average CPU time wastage. Average CPU time wastage is calculated as ratio of total CPU time wastage and total number of processor ($M$). CPU time wastage at any time slot is the number of number of free processor at that time unit. Wastage of CPU time happens because of these following reasons.
\begin{enumerate}
\item Some processors may be free at one time slot because the remaining processors are not sufficient for any ready task (processor requirement of ready tasks is higher than the available free processors).
\item Some processors may be free if the total requirement of all ready tasks is less than the total available processors at any one of the time slot.
\end{enumerate} 

\section{Previous Work} \label{prevwork}

\subsection{Scheduling task on multiprocessor}

As described in Section 1, non-preemptive scheduling of independent tasks on $m \geq 3$ processors is NP-complete. Also, the problem of scheduling a finite set of tasks having some precedence constraint on finite set of multiprocessor with goal of minimizing makespan is NP-complete for most of the cases except for a few simplified cases. Many heuristics with polynomial-time complexity have been suggested based on their assumptions about the structure of the parallel program and the target parallel architecture \cite{7}. These assumptions includes (a) uniform task execution times, (b) zero inter-task communication times, (c) contention-free communication, (d) full connectivity of parallel processors, and (e) availability of unlimited number of processors.
\\However these assumptions may not hold in real world for a number of reasons. Even after making above assumptions, scheduling problem is NP-complete in these following cases \cite{7} : (a) scheduling tasks with uniform weights to an arbitrary number of processors and (b) scheduling tasks with weights equal to one or two units to two processors.

As stated in \cite{7} there are only three special cases for which there exist optimal polynomial time algorithms. These cases are (a) scheduling tree-structured task graphs with uniform node weights on arbitrary number of processors in linear time by Hu's \cite{10} highest level first heuristics, (b) scheduling arbitrary task graphs with uniform node weights on two processors in quadratic time by Graham et. al.  \cite{9}, (c) scheduling an interval ordered task graph with uniform node weights to an arbitrary number of processors have been solved in linear time by Papadimitriou et. al. \cite{6}. However, even in these cases, communication among tasks of the parallel program is assumed to take zero time.  In \cite{12}, Ullman proved that DAG scheduling problems where considered DAG's nodes have unit weights and system has $m$ processors are NP-complete. He also proved that DAG scheduling problem where nodes have either one or two as a weight value and system has two processors is also NP-complete. Figure \ref{table:1}, shows complexity of scheduling problems  without communication time between tasks in tabular form. 

\begin{figure}[tb]
\centering
\footnotesize
\setlength{\tabcolsep}{0.6mm}
 \begin{tabular}{|c|c|c|c|c|c|c|}
 \hline
  \multirow{3}{*}{Structure} & \multicolumn{3}{c}{Uniprocessor Tasks} & \multicolumn{3}{|c|}{Unit time multiprocessor tasks} \\  
  \multirow{3}{*} & \multicolumn{3}{c}{(Processing time)} & \multicolumn{3}{|c|}{(Processor Requirement)} \\ 
  \cline{2-7}
  \multirow{3}{*} & Unit & Arbitrary & Arbitrary & Unit & Arbitrary & Arbitrary\\
   \multirow{3}{*} &  & (Non-pre & (Preem- &  & (Non-spl & (Split\\
   \multirow{3}{*} &  &  emptive) & ptive) &  &  itable) & able)\\
  \hline
  DAG 	 & NPC		  & NPC & NPC & NPC & NPC &  OPEN\\ 
  	 & Ullman\cite{12} & \cite{12} & \cite{25} & \cite{25} &  Bz\cite{Blazewicz13} &  \\ \hline
  Tree 	 & $O(n)$	 & NPC &  $n \log n $ & $O(n)$ & NPC &  OPEN \\ 
  	 & Hu\cite{Blazewicz13} 	 & \cite{23} & \cite{4}  & Hu & Bz\cite{Blazewicz13} &  \\ \hline
  Chain & $O(n)$ & NPC & $n+m\log n$& $O(n)$ & NPC&  OPEN\\ 
        &  Hu \cite{Blazewicz13} & \cite{25} & \cite{4} & Hu\cite{Blazewicz13} & Bz\cite{Blazewicz13} &  \\ \hline
  Indepe- & $O(n)$ & NPC & $O(n)$ & $O(n)$ & NPC, & NPC\\ 
    ndent &        &     &  \cite{25}    &  & Bz\cite{5} & Bz\cite{5} \\ \hline
 \end{tabular}
\label{table:1}
\caption{Complexity of scheduling problems on P processors without communication time between tasks}
\end{figure}

\subsection{Scheduling multiprocessor task on multiprocessor}

The problem of scheduling multiprocessor tasks on multiprocessor is even more harder. Blazewicz $et~al.$ \cite{Blazewicz13},  proposed $O(n\log n)$ time algorithm for chained multiprocessor tasks where processor requirement is uniform for each task in a chain, where $n$ is number of task  and $m$ is number of processor. They proposed $O(n\log n)$ algorithm for same type of applications with processor requirement in either increasing or decreasing fashion in a chain but having only two types of tasks either requiring $1$ processor or $k$ processor. Gonzalez $et~al.$ \cite{4},  proposed a polynomial time preemptive algorithm for scheduling trees in $O(n\log m)$ time in off-line mode. Algorithm given by them solves forests of $n$ tasks onto $m$ identical processors by minimizing the number of preemption in worst case. Blazewicz $et~al.$ \cite{2}, proposed scheduling algorithm for independent processor tasks. They divided the tasks into two sets $t-type$ sets and $w-type$ sets. $t$-type tasks are those tasks which require one arbitrary processor for execution and $w$-type are those which require two arbitrary processors. The non preemptive version of scheduling $t$ and $w$ types tasks is NP-complete but in this preemptive version is given which schedules $t$ and $w$ tasks in $O(n\log m)$ time. Blazewicz $et~al.$ \cite{5} proposed linear time algorithm for scheduling tasks requiring an arbitrary number of processors between $1$ and $k$, where $k$ is fixed integer and unit time processing. If $k$ is not fixed than problem is NP-complete. They considered both preemptive and non preemptive versions for two types of problem sets. One is if there are only two types of tasks in the set, one requiring $1$ processor and other requiring $k$ processors then both non preemptive version and preemptive version take $O(n)$ time to schedule the tasks. Second is if tasks require arbitrary number of processors between $1$ and $k$ then non preemptive version gives complexity O$(m^{k-1}.n^k)$ and preemptive version gives complexity of $O(n^m)$. Li $et~al.$ \cite{8}, proposed a task duplication based scheduling algorithm for fork-join task graph with complexity of $O(n^2)$. Blazewicz $et~al.$ \cite{3},  proposed a scheduling algorithm for two-processors tasks on uniform $2$-processor system and schedule $w$ and $t$ types tasks in $O(nm+n\log n)$ time. 

\section{Uniform Chains of  Multiprocessor Unit Time Tasks} \label{unichain}

\subsection{Splitable multiprocessor tasks}  \label{secunipitable}
In this section, we propose an optimal algorithm longest chain maximum processor first (LCMPF) for the minimization of makespan using simple rules. As shown in algorithm \ref{algo1} it works on two criteria, first one is length of chain  and second is processor occupancy of task. The application which has longest chain means maximum number of phases (application length) will first scheduled. If two or more applications or chains are of same length then the application which have more number of processor requirement (initial requirement) of task will be scheduled first. Figure \ref{figuniformex} shows an example of uniform chain multiprocessor unit time task system. Initial processor requirement of chain $C_1$, $C_2$, $C_3$ and $C_4$ is $p_{1}$, $p_{2}$, $p_3$ and $p_4$ respectively. Total number of phases or tasks in $C_1$, $C_2$, $C_3$ and $C_4$ is $4$, $3$, $5$ and $4$ respectively. 

\begin{algorithm}[tb]
\footnotesize	
\textbf{Input:} Set of $N$ Application Chains and $M$ Processors.  
\begin{algorithmic}[1]
\WHILE {All chains are not scheduled completely}
\STATE Sort chains in the non-increasing order of remaining unscheduled chain length.
\WHILE{at least one processor is free}
\STATE \textbf{if} {there are two or more chains of same unscheduled length} \textbf{then}
    \STATE $\:\:\:$  Select next ready task $T_{ij}$ from that chain which has max. proc. req. (initial req.) in all chains of same length.
\STATE \textbf{else} Select next ready task $T_{ij}$ from chain with longest unscheduled length.
\STATE \textbf{if} {Remaining processors m are more than $p_{ij}$ of selected task $T_{ij}$} \textbf{then}
\STATE  $\:\:\:$ Schedule task $T_{ij}$ on Allocated Processors	
\STATE \textbf{else} Schedule task $T_{ij}$ on remaining processors m and make this task as next ready task for current chain with proc. req. equal to $p_{ij}$-m.
\ENDWHILE     		
\ENDWHILE
\end{algorithmic}
\caption{\textbf{Longest Chain Maximum Processor First (LCMPF)}}
\label{algo1}
\end{algorithm}

\begin{figure}[tb]%
    \centering
    \subfigure[Uniform chains of multiprocessor unit time task]{{\label{figuniformex}\includegraphics[scale =0.38]{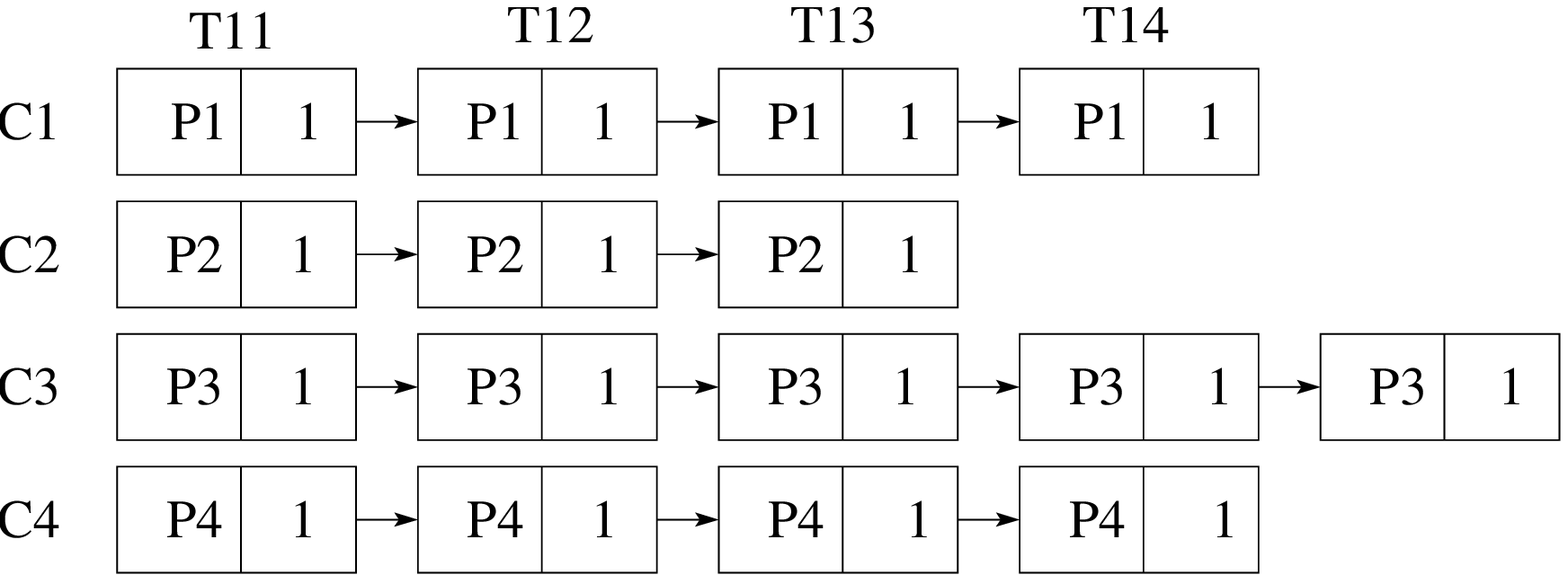} }}%
\hspace{2mm}
    \subfigure[Assuming some values: applications set at time = 1]{{\label{figuniformexval}\includegraphics[scale = 0.38]{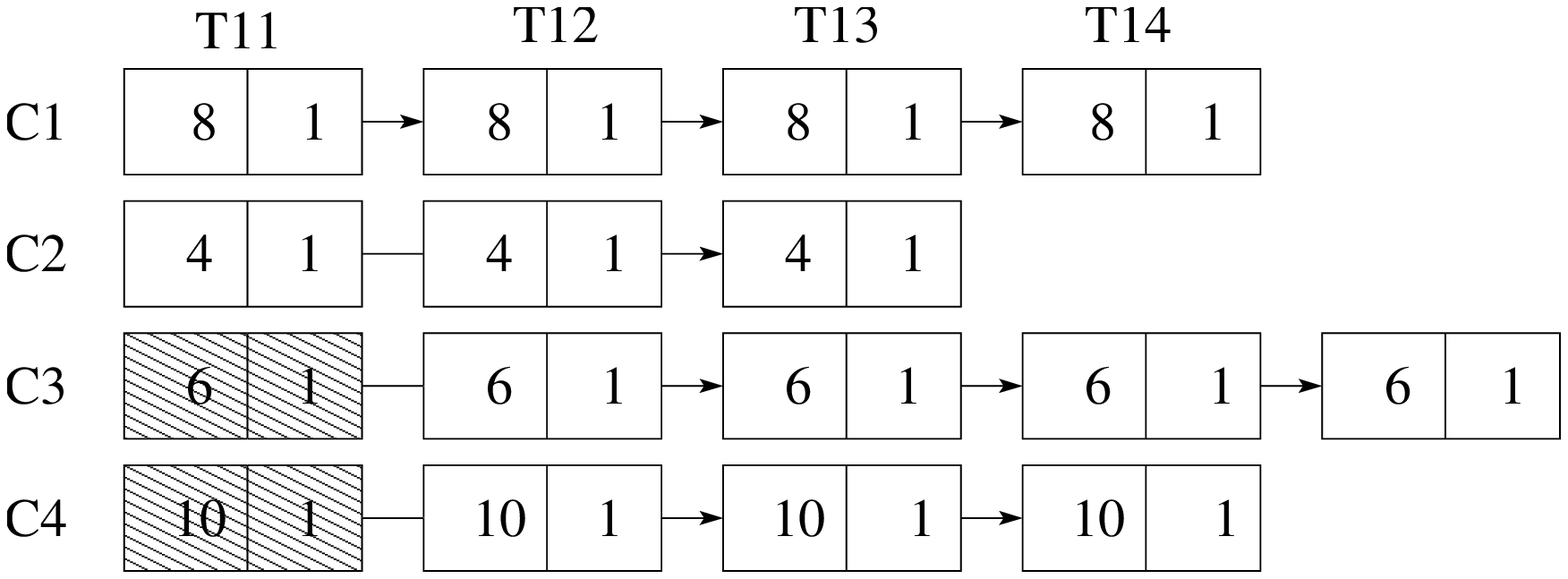} }}%
\caption{Example of uniform chains of multiprocessor unit time task}
    \label{figasrbitrary}%
\end{figure}

Pseudocode for LCMPF is shown in Algorithm \ref{algo1}. In this example, assuming M = 16, $p_1$ = 8, $p_2$ = 4, $p_3$ = 6 and $p_4$ = 10, uniform chain task system of multiprocessor unit time tasks is shown in Figure \ref{figuniformexval}.  algorithm works as follows: initially at time $t = 0$, chains in sorted order based on remaining unscheduled length are $C_3$, $C_1$, $C_4$ and $C_2$ and number of free processors is $16$. So task $T_{31}$ will be scheduled first. Now, number of remaining free processors is $m = 16-6 = 10$. Now, next longest chains are $C_1$ and $C_4$ but initial processor requirement of $C_4$ (i.e. $p_4$ = 10) is greater than initial processor requirement of $C_1$ (i.e. $p_1$ = 8), so next task to be scheduled will be $T_{41}$. Now, number of remaining free processors is $m = 10-10 = 0$ so no task can be scheduled next in this time slot. As there is no free processor so CPU time wastage in current time slot ($t = 0$) is 0.     

Time complexity of this can be analyzed as follows:  Sorting chains requires $O(N\log N)$ time and finding chain which requires maximum initial processors from longest chains of same length takes $O(N)$ time. For one time slot, scheduling will take time O($N+N\log N$) time, so total time complexity of scheduling task system by LCMPF algorithm is $C_{max}$.$(N+N\log N)$, where $C_{max}$ is constant.

Let us assume length of chain have an upper bound $L$. In this case, instead of sorting we can use $L$ different bins to put the application with length $l$ to bin number $l$ where highest bin is $L$, next highest bin is $L-1$ and so on. This will take $O(N)$ time to put all the applications to bins. If we want all applications in each bin to be in sorted order based on initial processor requirement, then process of inserting all applications to bins requires $O(N\log N)$ time.

Every time slot, we need to take out some applications from the current highest bin ($L^{'}$) and get inserted to next highest bin ($L^{'}-1$). If we want applications in the bins in sorted order based on initial processor requirement then it will take $O(N)$ time. Let $\alpha$ be the number of applications get executed in a time slot, then we need to do maximum $\alpha$ applications to be removed from current highest bin ($L^{'}$) to next current highest bin $(L^{'}-1)$. So time complexity of this operation will be $O(\alpha+\beta)$, where $\beta$ is the number of applications already in next current highest bin $(L^{'}-1)$. Overall time complexity of algorithm will be $O(N\log N)+C_{max}.(\alpha+\beta)$.   

\begin{theorem}
Longest chain maximum processor first algorithm always gives optimal makespan time for uniform chains with splitable tasks.
\end{theorem}
\begin{proof}
We always try to use all the processors at all the time slots to make $p_{waste}$ minimum. CPU time wastage happens when the total number of processor requirement of all the ready tasks in a particular time slot is less than $M$. Optimality of LCMPF is shown as result of proof of Lemma \ref{lemma4.2} and Lemma \ref{lemma4.3}.
\end{proof}

\begin{lemma} \label{lemma4.2}
Selecting tasks from long chain first will not increase the CPU wastage time in future time slot.
\end{lemma}

\begin{proof}
Selecting long chain task first reduces the chances of free processors in future time slots. Let there be an algorithm $A^{'}$ which gives the optimal makespan time $T^{'}$ and our approach LCMPF is giving makespan time $T$.  As from equation (1) lower bound of makespan $LB$ will be same in both algorithms. As $A^{'}$ is optimal, we can say that
\begin{equation}
p_{waste}(A^{'}) \leq  p_{waste}(LCMPF)
\end{equation}
where $p_{waste}(A^{'})$ is average CPU time wastage of algorithm $A^{'}$ and   $p_{waste}(LCMPF)$ is average CPU time wastage of algorithm LCMPF.

If our approach selects task from longest chain (let the length of chain is $l$) then the optimal algorithm $A^{'}$ select task from any chain of length $l^{'}$ where $l^{'} \leq l$. This will happen each time and there will be 0 or more number of tasks of an arbitrary application remaining at last. Assume that two tasks of only application are remaining at last then the CPU wastage will increase because we can complete only one task at one time. 
\begin{figure}[tb]
\centering
\includegraphics[scale=0.43]{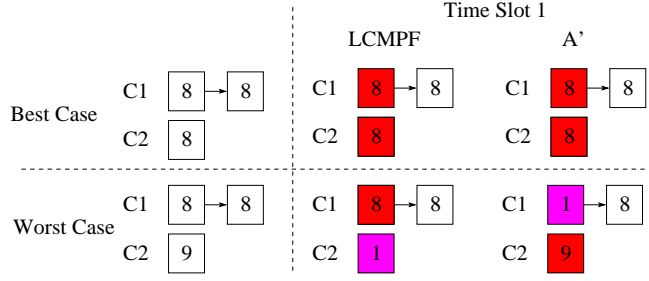} 
\caption{\textbf{Best and worst case behavior of algorithm $A^{'}$ and LCMPF}}
\label{figure16}
\end{figure}

The figure \ref{figure16} depict the worst case behavior. Assuming $16$ processors, in best case LCMPF algorithm selects $T_{11}$ first and $A^{'}$ selects either $T_{11}$ or $T_{21}$ first. In both cases total CPU time wastage will be $16-8 = 8$. In worst case LCMPF selects $T_{11}$ first so total CPU time wastage will be 16-(1+8) = $7$ only but if $A^{'}$ selects $T_{21}$ first then total CPU time wastage will be $(16-1)+(16-8) = 15+8 = 23$. So in last worst case there will be situation in $A^{'}$ that two phases of an arbitrary application will remain unscheduled i.e.
\begin{equation}
p_{waste}(A^{'}) \geq  p_{waste}(LCMPF)
\end{equation}
So by contradiction, we can say that selecting tasks from long chain first reduces the CPU time wastage.
\end{proof} 
\begin{lemma} \label{lemma4.3} 
Selecting task from highest number of processors requiring (initial requirement) chain will not increase the CPU time wastage of current time slot.
\end{lemma}
\begin{proof}
Let there be an optimal algorithm $A^{'}$ which gives the optimal makespan time then it will choose the task from the chain which requires initial processors at least one less than initial processors of task chosen by our approach LCMPF. If $A^{'}$ is occupying less processors then obviously $A^{'}$ is giving more CPU time wastage then LCMPF.
\end{proof}

Our approach LCMPF gives minimum CPU time wastage so it gives optimal makespan time. 

\subsection{Non-splitable multiprocessor tasks}

In this case, our application system consists of N uniform chains of non-splitable multiprocessor unit time tasks. Non-splitable multiprocessor task means task can be processed only by giving all the required number of processor as a whole at that time.

In the Algorithm \ref{algo2}, we are proposing an optimal algorithm for the minimization of make-span. This algorithm is similar to Algorithm \ref{algo1} except it does not use the second criteria (i.e. maximum processor occupancy). If we use LCMPF then also makespan time will be same for uniform chains having  non-splitable multiprocessor unit time  tasks but time complexity will increase because of second criteria. As soon as we found longest unvisited chain we select the next ready task of that chain, no matter if there are other chains of same length with high or less number of processor requirement. 
\begin{algorithm}[tb]
\footnotesize
\textbf{Input:} Set of $N$ Application Chains and $M$ Processors. 
\begin{algorithmic}[1]
\WHILE {All chains are not scheduled completely}
\STATE Sort chains in the non-increasing order of remaining unscheduled chain length.
\WHILE{Processors are free and at least one chain in unvisited}
\STATE \textbf{if} {only one unvisited chains of same unscheduled length} \textbf{then}
\STATE $\:\:\:$ Select ready task $T_{ij}$ from unvisited chain with longest unscheduled length.
\STATE \textbf{else} Select ready task $T_{ij}$ from any chain (or, FCFS basis).
\STATE \textbf{if} {Remaining processors m are more than $p_{ij}$}\textbf{then}
\STATE $\:\:\:$  Schedule task $T_{ij}$ on Allocated Processors
\STATE \textbf{else}  Mark this chain as visited;	
\ENDWHILE
\STATE Mark all chains as unvisited. $T_{sch}=T_{sch}+1$
\ENDWHILE
\end{algorithmic}
\caption{\textbf{Longest Chain First (LCF)}}
\label{algo2}
\end{algorithm}

To demonstrate the working of our scheduling approach, let us consider example of same set of applications described in Subsection \ref{secunipitable} and it is shown in Figure \ref{figuniformex}. Pseudo code for proposed longest chain first (LCF) approach is shown in Algorithm \ref{algo2}. The algorithm works as follows: initially at time $t = 0$, chains in sorted order based on remaining unscheduled length are $C_3$, $C_1$, $C_4$ and $C_2$ and number of free processors is $16$. So task $T_{31}$ will first  scheduled. Now, number of remaining free processors is $m = 16-6 = 10$ and next longest unvisited chains are $C_1$ and $C_4$. We can schedule ready task from any chain but we will schedule according to first come first served basis. It will not affect the makespan time, only scheduling order so next task to be scheduled will be $T_{11}$. Now, number of remaining free processors is $m = 10-8 = 2$. Next longest unvisited chain is $C_4$ but it has more processor requirement (i.e. $p_4$ = 10) than remaining free processors (i.e. $m$ = 2). Same problem is with $C_2$. Now, all the chains are visited so we can't schedule any task in current time slot. CPU time wastage in current time slot ($t = 0$) is 2.

If we assume maximum length of chain may be arbitrary, then complexity will be $C_{max}$ times the sorting time that is $C_{max}.(N\log N)$. But if length of chains are bounded by $L$ then instead of sorting we can use $L$ different bins to put the application with length $'l'$ to bin number $'l'$ where highest bin is $L$, next highest bin is $L-1$ and so on. This will take $O(N)$ time to put all applications to bins.

Every time slot, we need to take out some applications from the current highest bin ($L^{'}$) and get inserted to next highest bin ($L^{'}-1$). Let $\alpha$ be the number of applications get executed in a time slot, then we need to do maximum $\alpha$ applications to be removed from current highest bin ($L^{'}$) to next current highest bin $(L^{'}-1)$. So time complexity of this operation will be $\alpha$. Overall time complexity of algorithm will be $O(N)+\alpha.C_{max}$.   

\begin{theorem} \label{Theorem2}
Longest chain first algorithm always gives optimal makespan time for uniform chains with non-splitable multiprocessor unit time tasks.
\end{theorem}

\begin{proof} Let there be an optimal algorithm $A^{'}$ which gives schedule length $C_{max} = OPT$. Suppose the optimal algorithm produces a sequence of scheduled multiprocessor tasks at time slot $t_i$, where $t = 1$ to $C_{max}$. Suppose there are $N$ uniform chains of non-splitable multiprocessor tasks namely $a-chain$, $b-chain$, $c-chain$, $d-chain$ and etc. with processor requirement $a$, $b$, $c$, $d$ and etc. as shown in Figure \ref{figureS}.
\begin{figure}[tb]
\centering
\includegraphics[scale=0.43]{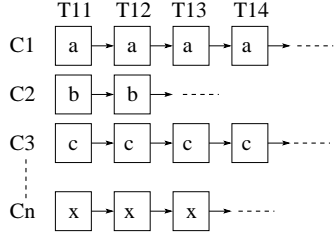} 
\caption{\textbf{Uniform chains of non-splitable multiprocessor unit time tasks}}
\label{figureS}
\end{figure} 
The optimal algorithm produces a output sequence of scheduled tasks for example $1(a,b)$, $2(a,c)$, $3(c,d,a)$, $4(a,d)$.........., where $i(j,k,..)$ represents a task from $j$, a task from $k$ and so on are scheduled at time slot $t = i$. We can see from the scheduled sequence that between a nearby pair, scheduled slot may be exchanged. For an example, $1(a,b)$ and $2(a,c)$ can be exchanged without affecting the optimality and precedence constraint to $2(a,b)$ and $1(a,c)$. So in this way, we can sort the scheduling sequence by $A^{'}$ to get the required scheduling sequence by LCF.
\end{proof}

As optimal sequence can be converted to LCF sequence without changing $C_{max}$, so our LCF is optimal. But we can't guarantee that LCF is the only algorithm which produces optimal result. There may be other algorithm that may produce optimal result.

\section{Monotone chains of  multiprocessor unit time tasks} \label{monotone}

Non-increasing chains of multiprocessor unit time tasks can be scheduled optimally. Our proposed approaches LCMPF (Algorithm \ref{algo1}) and LCF (Algorithm \ref{algo2}) will produce optimal makespan time for non-increasing chains of  splitable and non-splitable   multiprocessor unit time tasks respectively. An example of non-increasing chains of multiprocessor unit time tasks is shown in Figure \ref{figdecreasing}. Optimality of LCMPF and LCF for non-increasing chains of  splitable and non-splitable multiprocessor unit time  tasks respectively can be proved by Theorem \ref{Theorem5.1} and Theorem \ref{Theorem5.2}. 
\begin{theorem} \label{Theorem5.1}
LCMPF algorithm schedules non-increasing chains of splitable multiprocessor unit time tasks optimally. 
\end{theorem}

\begin{proof}
As shown in Theorem \ref{Theorem2}, $P_{waste}$ by longest chain maximum processor first is minimum for uniform chains of multiprocessor unit time tasks. As we are choosing maximum processor first if chain lengths are same and this satisfy both uniform chains and non-increasing chains. So it is obvious that choosing maximum processor task of a chain will choose the task with highest processor requirement of the chain which is the ready task of the chain. As the subsequent tasks of chain require less number of processor then the $P_{waste}(LCMPF)$ of current slot will be minimum and also maintains the $C_{max}$ to be optimal.
\end{proof}

\begin{theorem} \label{Theorem5.1}
LCF algorithm schedules non-increasing chains of non-splitable multiprocessor unit time tasks optimally. 
\end{theorem}

\begin{proof} \label{Theorem5.2}
As shown in Lemma \ref{lemma4.2}, scheduled sequence given by an optimal algorithm $A^{'}$ can be converted into scheduled sequence given by LCF without affecting optimality of makespan in case of  uniform chains. Same can be done in case of  non-increasing chain without affecting the optimality. 
\end{proof}

An example of non-decreasing chains of multiprocessor unit time tasks is shown in Figure \ref{figincreasing}. This type of chain can also be scheduled for both types (splitable and non-spilitable) of task same as uniform chains. LCMPF and LCF will give optimal make-span time for  non-decreasing chains of splitable and non-splitable multiprocessor tasks respectively. Optimality of LCMPF and LCF for non-decreasing chains of  splitable and non-splitable multiprocessor unit time tasks respectively can be proved by Theorem \ref{Theorem5.3} and Theorem \ref{Theorem5.4}.

\begin{theorem}\label{Theorem5.3}
LCMPF algorithm schedules non-decreasing chains of splitable multiprocessor unit time tasks optimally.
\end{theorem}
\begin{proof}
Non-decreasing chains of multiprocessor unit time tasks can be formed by reversing the order of non-increasing chain. Reversing the chain order does not increase the chain length and it can be scheduled optimally by LCMPF. Non-decreasing chains are just opposite chains of non-increasing chains. Let the schedule order $D$ we get is $a_1$, $a_2$, $a_3$ and $a_4$ of non-decreasing chains. If we reverse the chains and schedule them as non-increasing chains then the order of tasks for scheduling we will get the opposite of $D$, and that is $a_4$, $a_3$, $a_2$ and $a_1$. As schedule for non-increasing chains is optimal, we can say that schedule for non-decreasing is also optimal.
\end{proof}

\begin{theorem} \label{Theorem5.4}
LCF algorithm schedules non-decreasing chains of non-splitable multiprocessor unit time tasks optimally.
\end{theorem}
\begin{proof}
As shown in Theorem \ref{Theorem5.3}, non-decreasing chain can be converted into non-increasing chains by reversing order. LCF gives optimal solution for non-increasing chains so we can say schedule by LCF for  non-decreasing chains of non-splitable multiprocessor unit time tasks will also be optimal. \ \\
\end{proof}

\section{Arbitrary Chains of multiprocessor unit time tasks} \label{arbichain} 

\subsection{Non-splitable tasks} \label{arbnonsplit}
In this we are considering system of application of non-splitable multiprocessor unit time tasks with arbitrary number of processor requirement in each phase of a chain. Unlike the processor requirement of  multiprocessor task of uniform chains or monotone chains where  are same or in non-increasing (or non-decreasing) order. Figure \ref{figarbitrary} shows an example of arbitrary chains of multiprocessor unit time  tasks. As the tasks are multiprocessor unit time tasks and non-splitable in terms of processor, this problem is proved to be a hard problem \cite{Blazewicz13}. 

We have used three heuristics for this problem and compared the performance. The proposed heuristics are (a) maximum criticality first (MCF), (b) longest chain maximum criticality first (LCMCF), and (c) longest chain maximum processor first (LCMPF). Longest chain maximum processor first heuristic is described in Section \ref{secunipitable}. 

As defined in \cite{10}, we level each multiprocessor task of chain with its criticality value, which is sum of processor requirement of self and all of its successors. The criticality of task $T_{ij}$ is calculated as 
\begin{equation}  \label{criticequ}
CV_{ij} = \sum_{k=j}^{k=n_i} p_{ik}
\end{equation}
where $CV_{ij}$ is the criticality of $j^{th}$ task of $i^{th}$ chain, $p_{ik}$ is processor requirement of task $T_{ik}$ and $n_i$ is number of phases or tasks of $i^{th}$ chain. Figure \ref{figure25} shows calculated criticality value of tasks of a example chain and now every task have one more parameter that is $CV_{ij}$.
\begin{figure}[tb]
\centering
\includegraphics[scale=0.43]{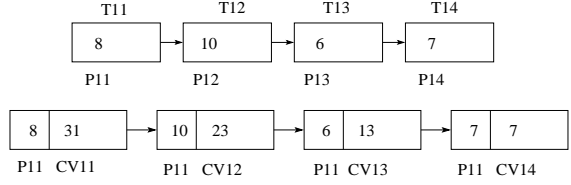} 
\caption{\textbf{Criticality of tasks of chain}}
\label{figure25}
\end{figure} 

 Pseudocode for MCF heuristic is shown in Algorithm \ref{algo3}. Initially, it calculates the criticality value (using equation 5) of tasks of all applications. In every scheduling step, it try to select subset of tasks from set of ready tasks to schedule that gives the maximum criticality value with total number of processors in the subset of ready tasks is less then $M$. If ready task of $i^{th}$ chain is $T_{ir}$ and its criticality value is $CV_{ir}$ and processor requirement is $p_{ir}$ then we need to select a subset $S$ from the ready tasks so that the following objective meets 
\begin{equation}
max(\sum_{T_{ir}\in S} CV_{ir}) \: \: \: with \:\:\:\sum_{T_{ir}\in S} p_{ir} \leq M 
\end{equation} 
This is similar to solving 0-1 knapsack problem.

\begin{algorithm}[tb]
\footnotesize
\textbf{Input:} Set of $N$ Application or Chains and $M$ Processors.  
\\\textbf{Output:} Schedule for application's tasks.
\begin{algorithmic}[1]
\STATE Calculate criticality for each task of each application.
\STATE Initialize ready queue with the first task of all applications
\WHILE {All chains are not scheduled or ready queue is not empty}
\STATE Select subset from ready tasks that gives maximum criticality and which has total processor requirement less than or equal to M.
\STATE Update ready queue with the ready tasks(i.e. whose previous tasks have been scheduled).
\ENDWHILE
\end{algorithmic}
\caption{\textbf{Maximum criticality first scheduling (MCF)}}
\label{algo3}
\end{algorithm}

Pseudo code for LCMCF is shown in Algorithm \ref{algo4}. It is similar to LCMPF except the difference between second criteria of selection in choosing ready tasks. LCMPF chooses the ready task with maximum processor while LCMCF chooses the task with maximum criticality.

\begin{algorithm}[tb]
\footnotesize
\textbf{Input:} Set of $N$ Application or Chains and $M$ Processors. 
\\\textbf{Output:} Schedule for application's tasks. 
\begin{algorithmic}[1]
\STATE Calculate criticality for each task of each application.
\WHILE {All chains are not scheduled}
\STATE Sort chains in the decreasing order of remaining unscheduled chain length.
\WHILE{Processors are free and at least one chain is unvisited}
\IF {There are two or more unvisited chains of same length}
    \STATE Select ready task $T_{ij}$ which has maximum criticality among all ready tasks of unvisited longest chains . 
    \IF{There are two or more tasks of same criticality}
     	\STATE Select any task.
     	\ENDIF
\ELSE
 \STATE Select ready task $T_{ij}$ from longest and unvisited chain. 
\ENDIF       
\IF {Remaining processors $m$ are more than $p_{ij}$}
    \STATE Schedule task $T_{ij}$ on allocated Processors
\ELSE \STATE Mark this chain as visited	
\ENDIF 
\ENDWHILE
\STATE Mark all chains as unvisited
\ENDWHILE
\end{algorithmic}
\caption{\textbf{Longest Chain Maximum Criticality First (LCMCF)}}
\label{algo4}
\end{algorithm}

Figure  \ref{figure30} shows the set of 5 arbitrary applications with non-splitable multiprocessor unit time tasks. Second parameter in each task of each chain is the criticality of that task calculated by equation \ref{criticequ}. Assuming 20 processors, lower bound (LB) for the example task system  will be $\simeq 8$. Scheduled outputs of all three proposed heuristics for set of applications (Figure \ref{figure30}) are shown in Figure \ref{tab2} in tabular form. For this set of applications, all three heuristics give same makespan time $C_{max} = 9$ and total CPU time wastage equals to 21 (and $p_{waste} = 21/20 = 1.05 \simeq 1$) but different set of tasks are scheduled in a time slot.

\begin{figure}[tb]
\centering
\includegraphics[scale=0.38]{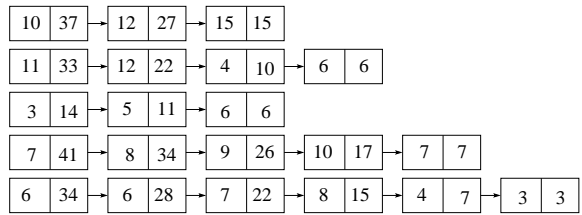} 
\caption{\textbf{Chain of non-splitable multiprocessor tasks with criticality values}}
\label{figure30}
\end{figure}   

\begin{figure}[tb]
\centering
\setlength{\tabcolsep}{1.1mm}
 \begin{tabular}{|c|l|l|l|}
  \hline
  Time & MCF output & LCMCF output & LCMPF output\\ \hline
  At time t = 0 & $T_{11}$, $T_{31}$, $T_{41}$ & $T_{31}$, $T_{41}$, $T_{51}$ & $T_{31}$, $T_{41}$, $T_{51}$ \\
  At time t = 1 & $T_{32}$, $T_{42}$, $T_{51}$ & $T_{32}$, $T_{42}$, $T_{52}$ & $T_{21}$, $T_{52}$ \\
  At time t = 2 & $T_{21}$, $T_{52}$ & $T_{21}$, $T_{53}$ & 
$T_{32}$, $T_{42}$, $T_{53}$ \\
  At time t = 3 & $T_{12}$, $T_{53}$ & $T_{11}$, $T_{43}$ &
$T_{22}$, $T_{54}$\\
  At time t = 4 & $T_{43}$, $T_{54}$ & $T_{22}$, $T_{54}$ &
$T_{32}$, $T_{42}$, $T_{52}$ \\
  At time t = 5 & $T_{33}$, $T_{44}$, $T_{55}$ & $T_{12}$, $T_{23}$, $T_{55}$ &
$T_{11}$, $T_{43}$ \\
  At time t = 6 & $T_{22}$, $T_{45}$ & $T_{44}$, $T_{24}$, $T_{56}$ & 
$T_{12}$, $T_{23}$, $T_{55}$  \\
  At time t = 7 & $T_{13}$, $T_{23}$ & $T_{13}$ &
$T_{24}$, $T_{44}$, $T_{56}$  \\
  At time t = 8 & $T_{24}$, $T_{56}$ & $T_{33}$, $T_{45}$ & 
$T_{23}$, $T_{45}$\\
  \hline 
 \end{tabular}
 \label{tab2}
 \caption{Scheduling results of all three heuristics of task system shown in Figure \ref{figure30} on 20 processors}
\end{figure}

\begin{figure}[tb]%
    \centering
    \subfigure[Cmax $L_{MCF}$=9, $L_{LCMCF}$=8 and $L_{LCMPF}$=8)]{{\label{arbit1}\includegraphics[scale =0.47]{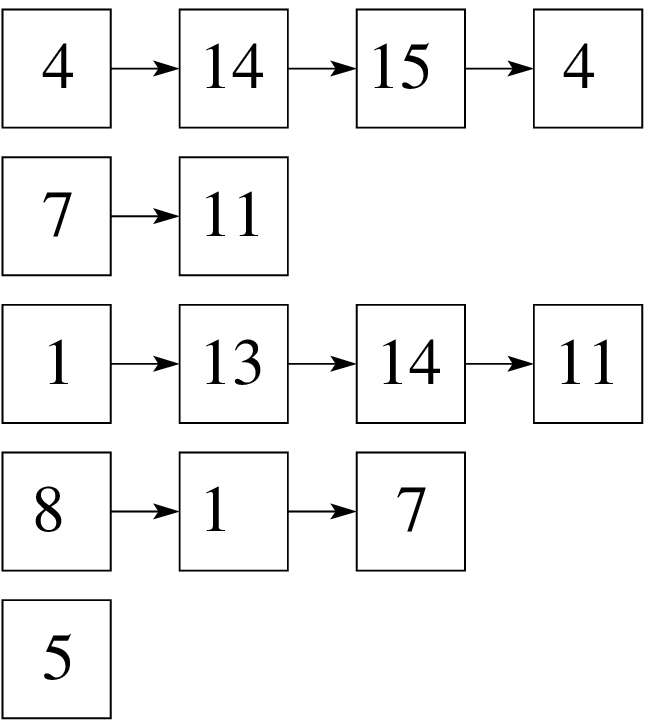} }}%
\hspace{2mm}
    \subfigure[Cmax $L_{MCF}$=9, $L_{LCMCF}$=10 and $L_{LCMPF}$=9]{{\label{arbit2}\includegraphics[scale = 0.47]{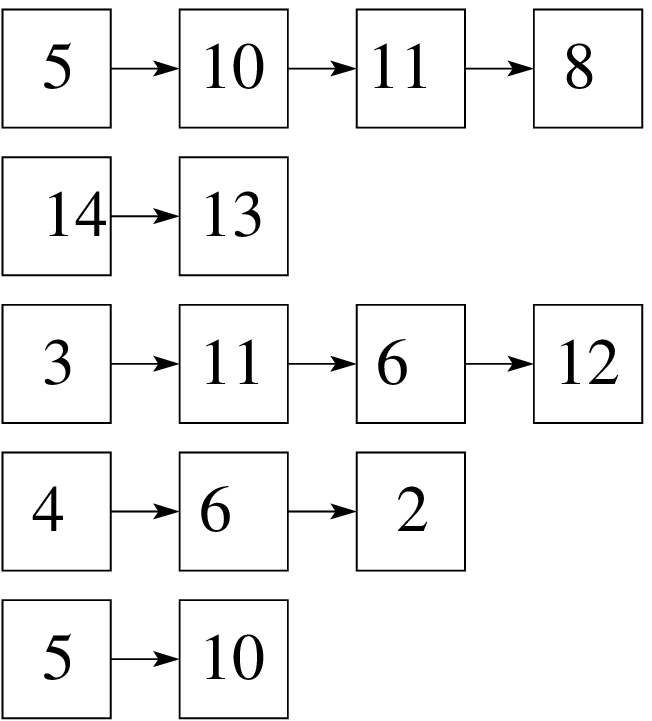} }}%
\hspace{2mm}    
    \subfigure[Cmax $L_{MCF}$=7, $L_{LCMCF}$=7 and $L_{LCMPF}$=8]{{\label{arbit3}\includegraphics[scale = 0.47]{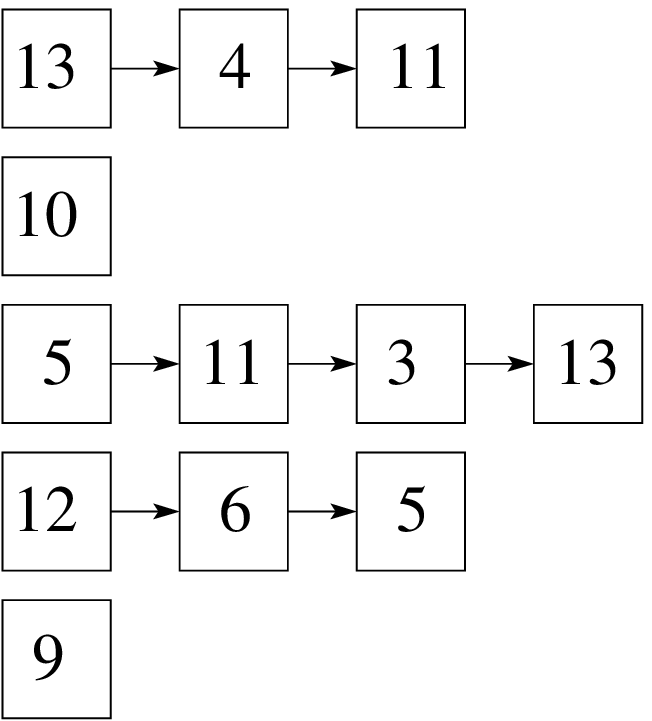} }}%
    \caption{Performance of MCF, LCMCF and LCMPF on contradictory examples}
    \label{fig:arbitrary}%
\end{figure}

\begin{figure}[tb]%
    \centering
    \subfigure[M = 64, Apps = 100]{{\label{a}\includegraphics[scale =0.28, angle = -90]{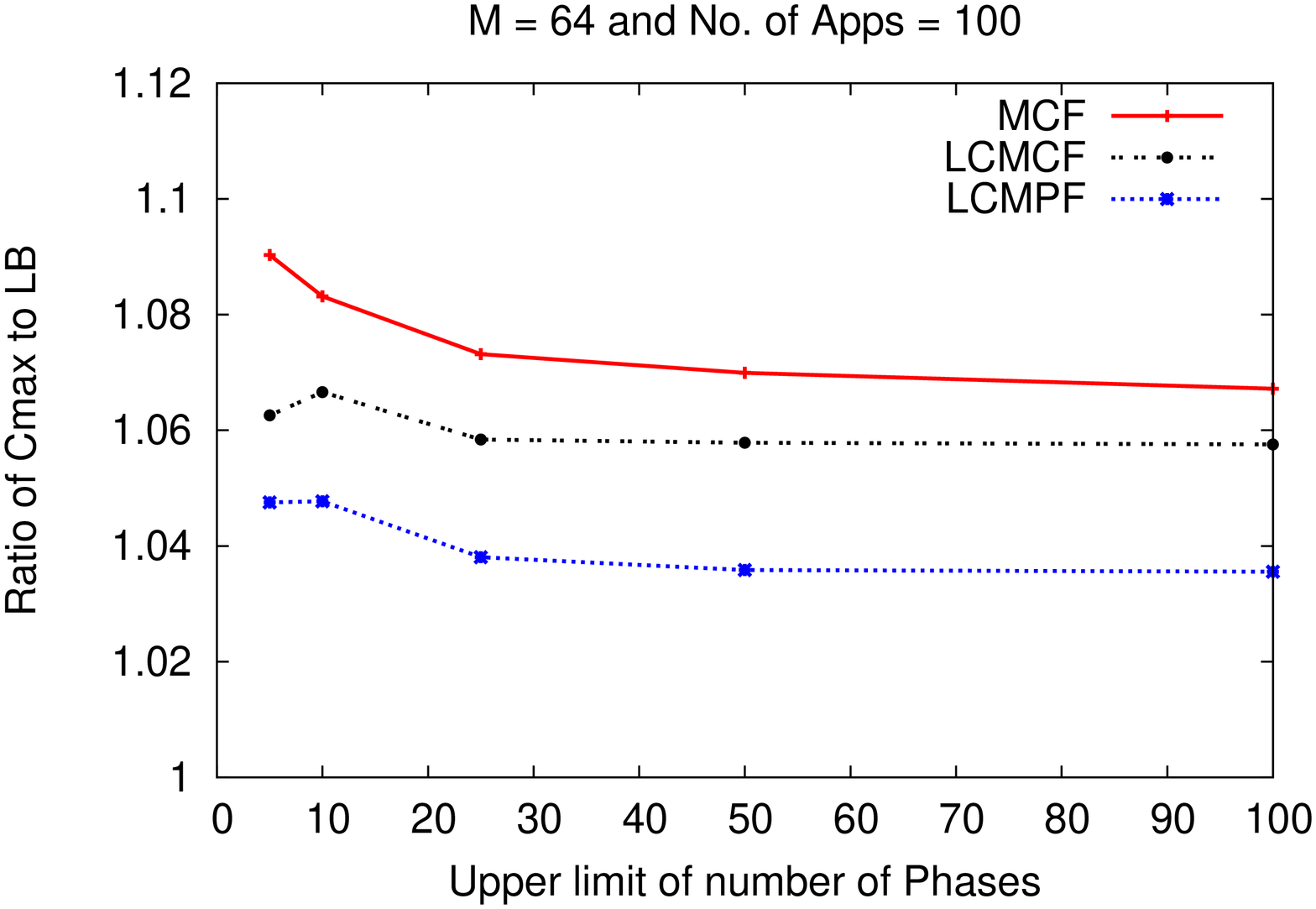} }}%
    \subfigure[M = 64, Apps = 1000]{{\label{c}\includegraphics[scale = 0.28, angle = -90]{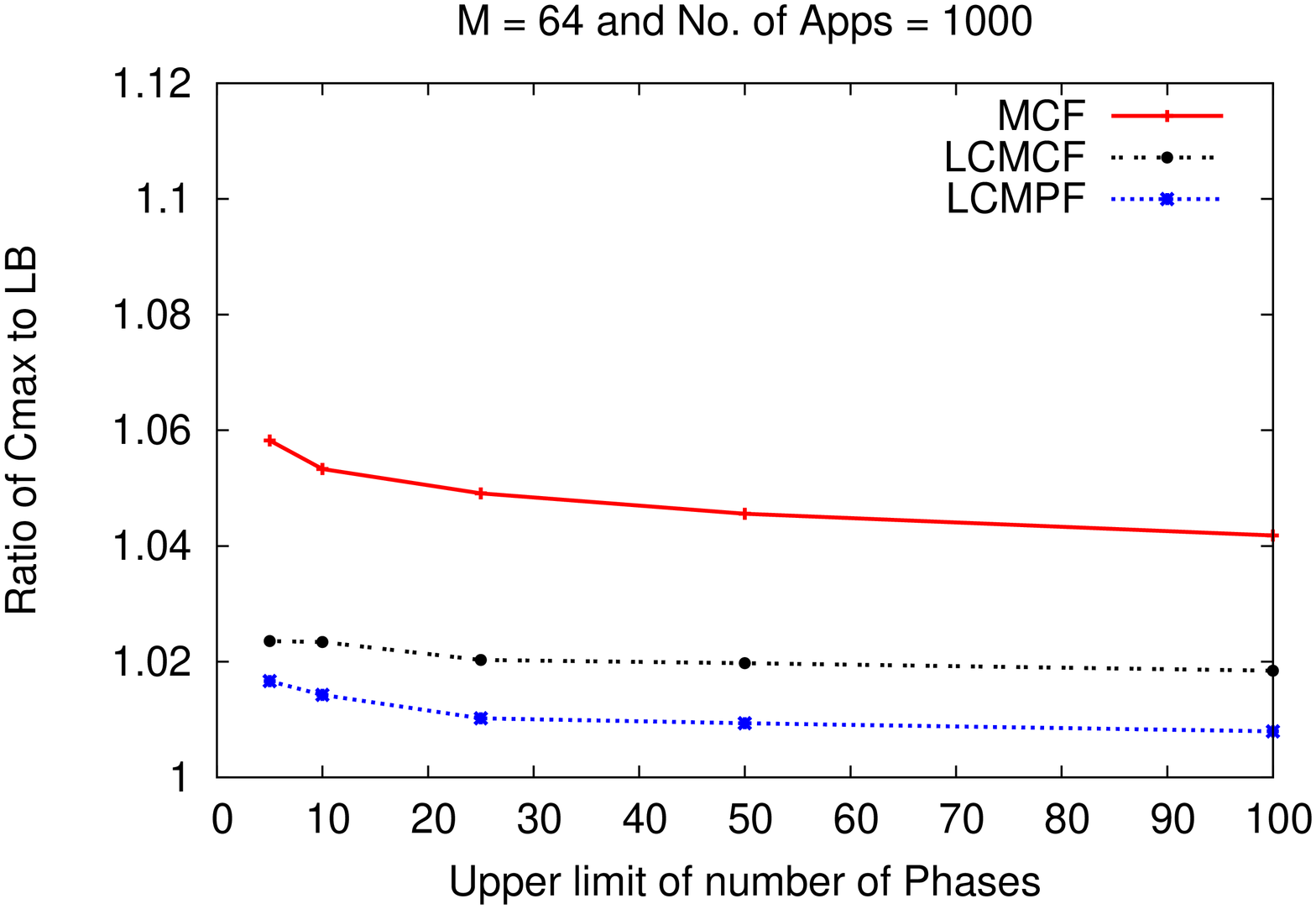} }}%
    \qquad
    \centering
    \subfigure[M = 512, Apps = 100]{{\label{d}\includegraphics[scale = 0.28, angle = -90]{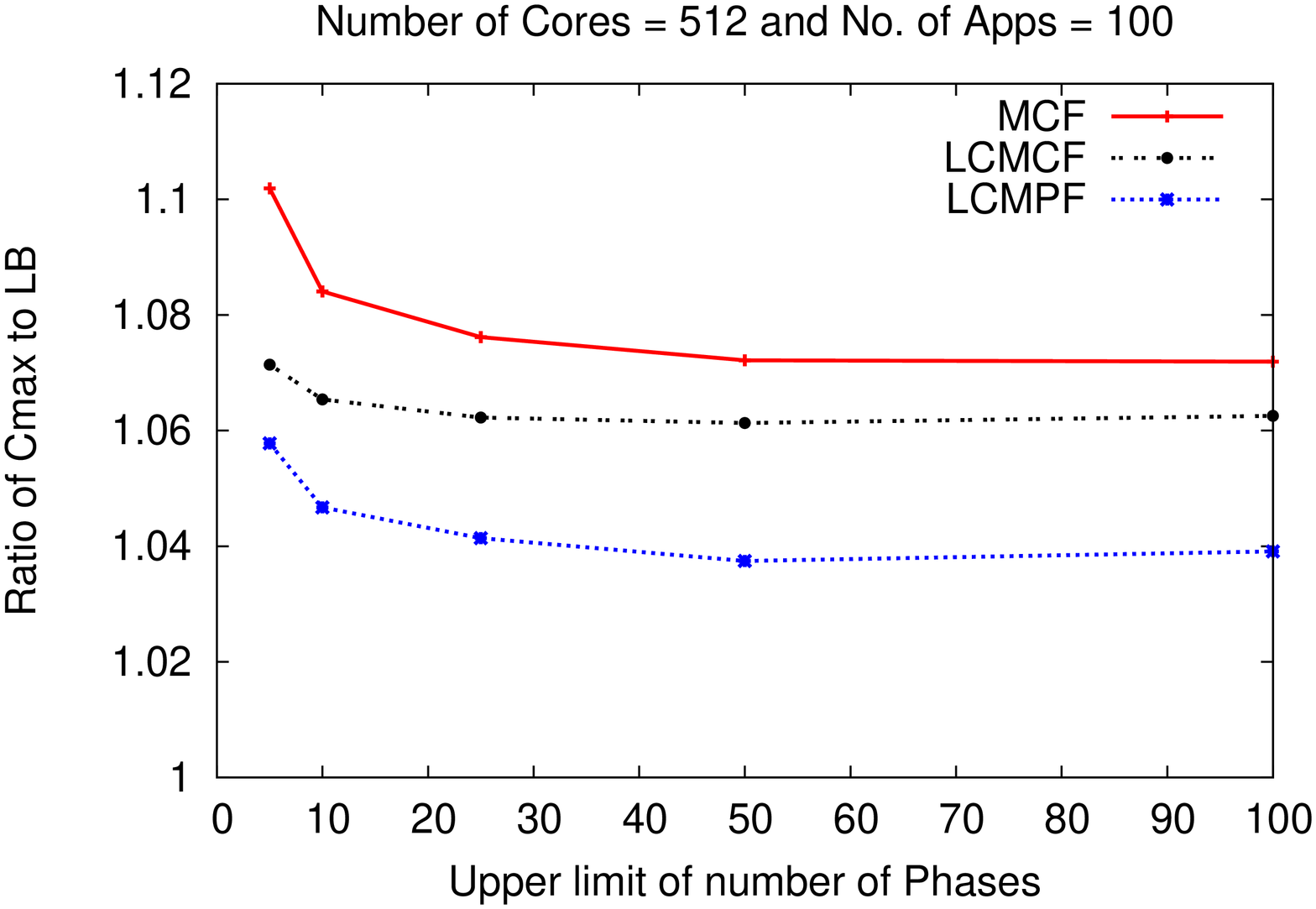} }}%
    \subfigure[M = 512, Apps = 1000]{{\label{f}\includegraphics[scale = 0.28, angle = -90]{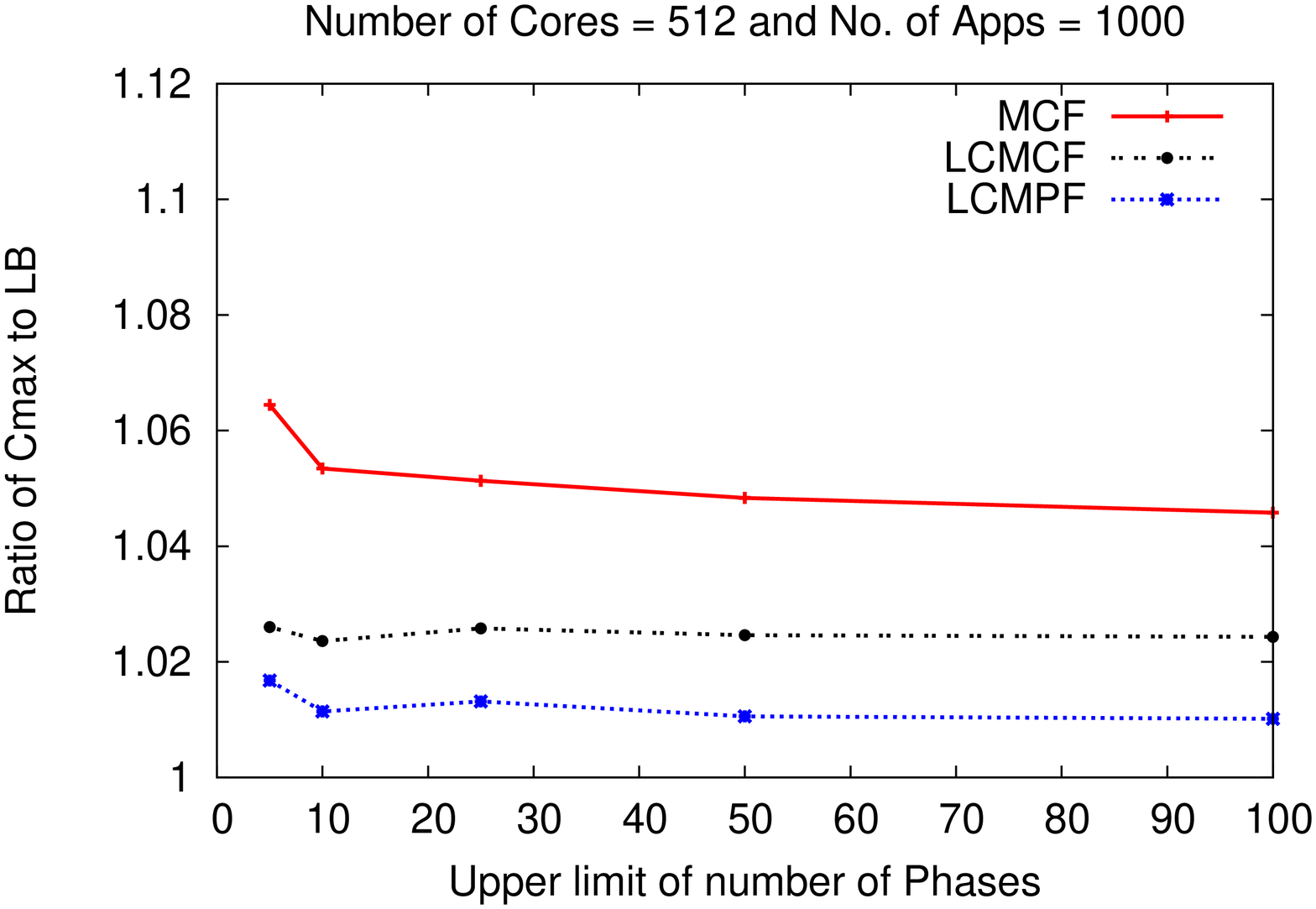} }}%
    \caption{\textbf{Comparison between MCF, LCMCF and LCMPF}}%
    \label{fig:example1}%
\end{figure}

\begin{figure}[tb]%
    \centering
    \subfigure[Phase variation = $\pm 10\%$]{{
    \label{g}\includegraphics[scale = 0.28, angle = -90]{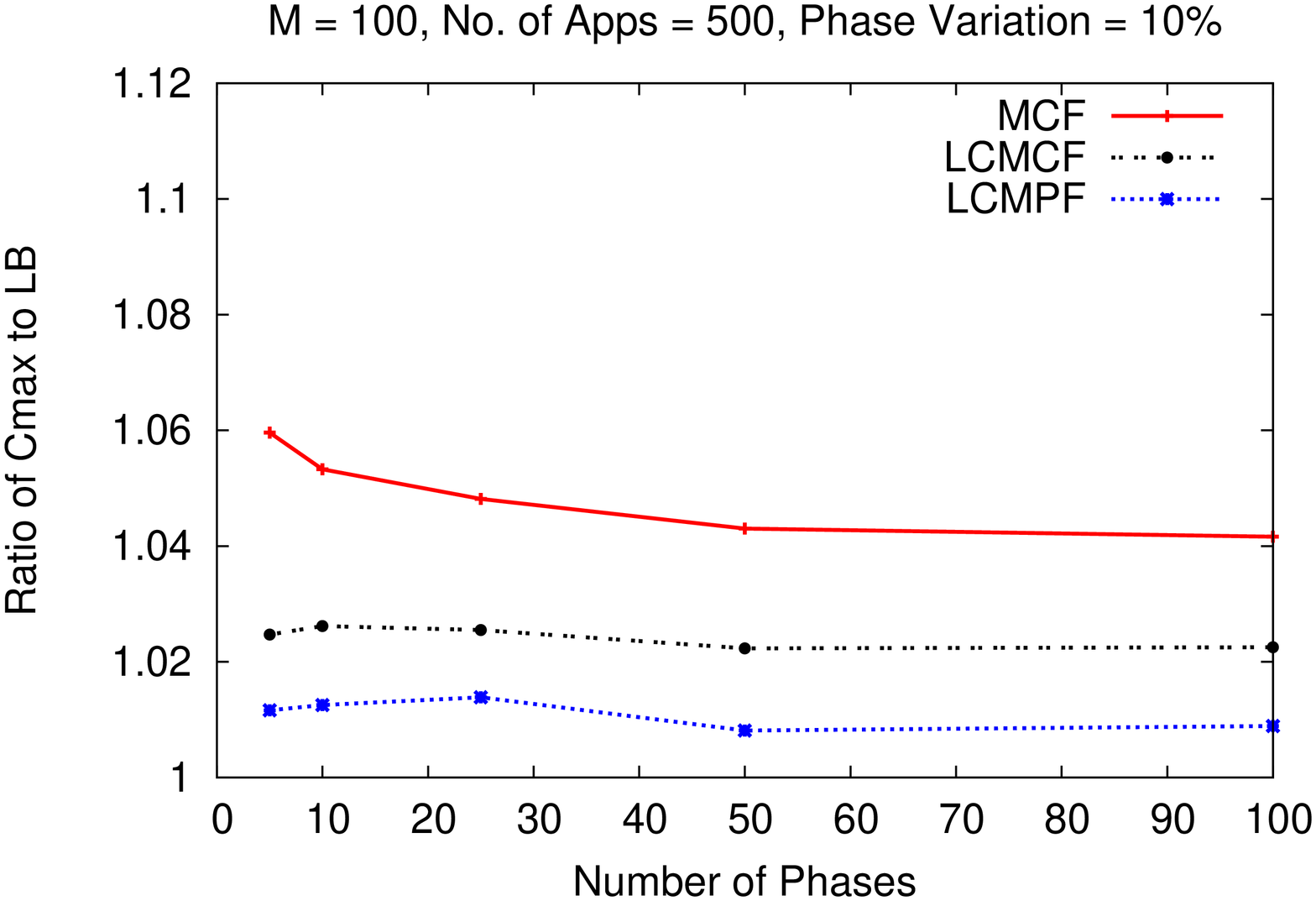} }}%
    \subfigure[Phase variation = $\pm 50\%$]{{
    \label{i}\includegraphics[scale = 0.28, angle = -90]{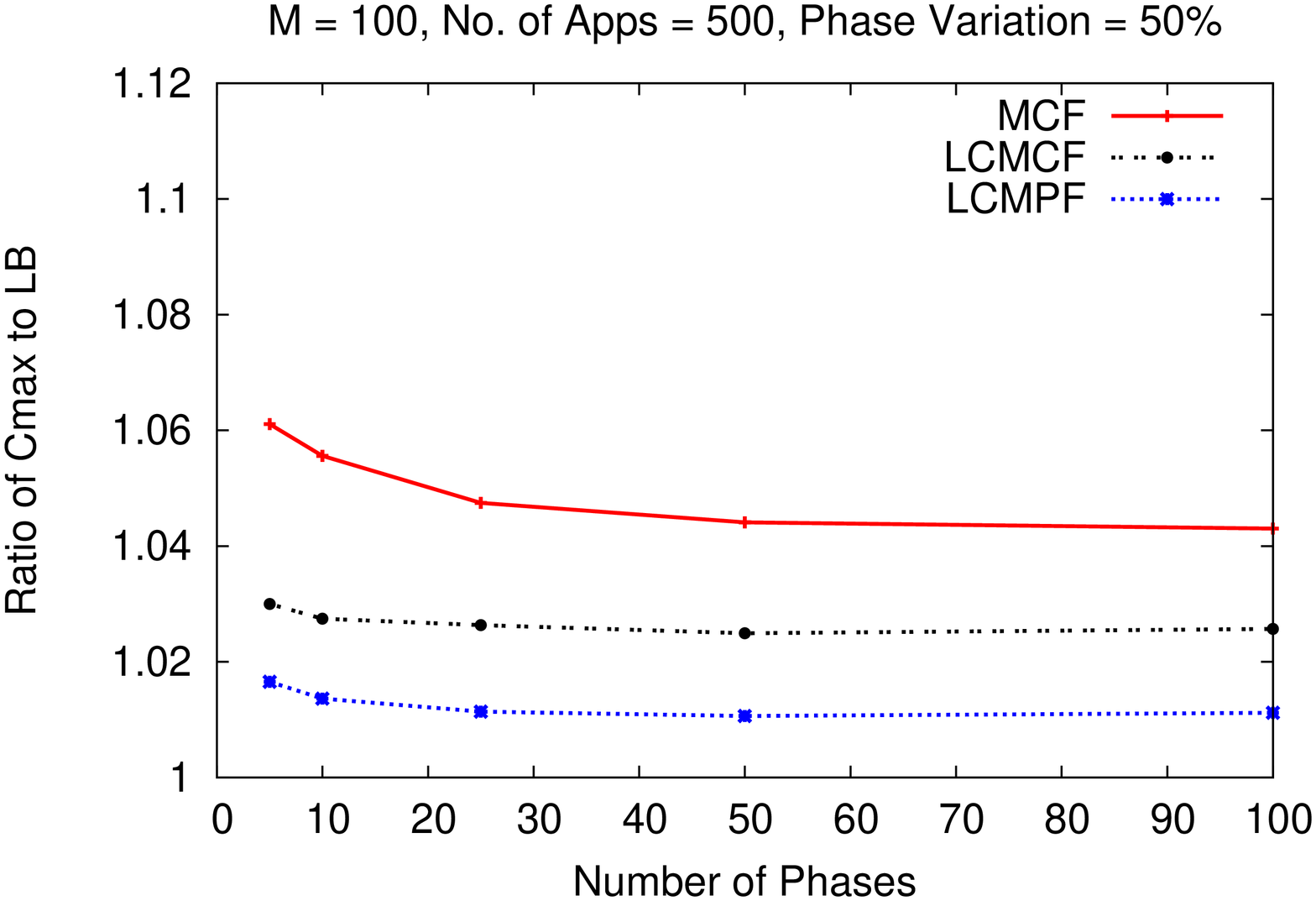} }}%
    \caption{\textbf{Comparison between MCF, LCMCF and LCMPF with variation of phases in different applications}}%
    \label{fig:example2}%
\end{figure}

Among these three heuristics, no one gives best result for all the case. One heuristic may perform better as compared to other two in some specific case (lower the $C_{max}$ value, better the result). As shown in Figure \ref{fig:arbitrary}, performance of LCMCF and LCMPF is better compare to MCF for application set 1 (shown in Figure \ref{fig:arbitrary}(a)), performance of MCF and LCMPF is better than LCMCF for applications set 2 (shown in figure \ref{fig:arbitrary}(b)) and performance of MCF and LCMCF is better than LCMPF for application set 3 (shown in Figure \ref{fig:arbitrary}(c)).

Figures \ref{fig:example1}a and \ref{fig:example1}b show the performance of MCF, LCMCF and LCMPF scheduling of 100 and 1000 applications respectively with different number of phases on 64 processor system.   Figures \ref{fig:example1}c and \ref{fig:example1}d show the performance of same scheduling approaches of same application set  on 512 processors.  We  observed that:
\begin{enumerate}
\item For a fixed number of processors, when number of applications increase, the difference between LB and makespan time decreases i.e. efficiency in time increases.
\item For fixed number of processors and applications, when upper limit of number of phases an application can have increase efficiency in time increases for all three heuristics.
\item Avgerage ratio of $C_{max}$ and $LB$ will always be better in LCMPF.       
\end{enumerate}

Again Figure \ref{fig:example2} shows the performance of MCF, LCMCF and LCMPF scheduling of application  system with  10\% and 50\% variation of number of phases of applications with one another for 500 applications on 100 processors. Our observation  says  that
\begin{enumerate}
\item For a fixed number of processors and applications, when phase variation, the difference between LB and makespan time varies i.e. efficiency in time increases or decreases.
\item For fixed number of processors, applications and phase variation, when number of phases an application can have increase efficiency in time increases for all three heuristics.
\end{enumerate}

Time complexity of MCF will be the complexity of 0-1 knapsack solution at each time slot. As problem of 0-1 knapsack can be solved in pseudo-polynomial time with respect to capacity of knapsack and total number of item, so in our case it will be $O(NM)$, so total complexity of MCF is $C_{max}.NM$. As described in Section 4, complexity of LCMCF will be same as complexity of LCMPF. 

Overall LCMPF will give better results than MCF and LCMCF for any set of applications having any number of phases and any number of applications.

\subsection{Splitable tasks}

Scheduling of arbitrary chains of splitable multiprocessor unit time tasks is an interesting problem. As of our knowledge, no one have found polynomial time solution and also no one has proved that this problem is NP-Complete. Using processor as continuous medium which behaves like electrical charge passing from task to task in the DAG (instead of chain), author of paper \cite{26}, solve this in iterative ways with complexity is $O(e^2+ne+I(n+e))$, where $e$ is the number of edges in the precedence graph and $I$  is the number of iterations in the algorithm. They use optimality conditions impose by a set of nonlinear equations on the  flow of processing power (processors) and on the completion times of independent paths of execution which is analogous to Kirchhoff's laws of electrical circuit theory. But our main aim is to solve in using discrete approach. 

We have also used the same MCF, LCMCF and LCMPF heuristics to schedule this kind of application on to multicore. As multiprocessor task are splitable, we use fractional knapsack in MCF heuristic. Experiment shows all three heuristics produce exactly same result for randomly generated examples. We also observe that all three heuristics perform equally if sum of processor requirement  of all the ready multiprocessor tasks is $ \geq M$ in all schedule time slot except  the last time slot. If this condition is violated then LCMPF performed better than other two heuristics.

\section{Conclusion and Future work} \label{concl}

Scheduling with considering the phase behavior  improve the system performance. Our proposed approach LCMPF scheduling of uniform and monotone chain of multiprocessor unit time task is proved to optimal. If the multiprocessor task are non-splitable, then LCF approach is optimal, we don't need to consider processor occupancy criteria of multiprocessor task.  

Scheduling arbitrary chain of multiprocessor unit time task is in NP-complete. In this case our proposed LCMPF based heuristics perform better as compared to MCF and LCMCF heuristics. We believe that scheduling of arbitrary chain of splitable multiprocessor unit time task is still an open problem. We have also compared  performance of proposed LCMPF and other MCF and LCMCF heuristics for scheduling this kind of task. In future, we are planning to try to solve scheduling of arbitrary chain of splitable multiprocessor unit time task. Also solve the same with other restrictive precedence constraints and or with some communication model.

\bibliographystyle{IEEEtran}


\end{document}